\documentclass[sigconf]{acmart}
\pdfoutput=1

\renewcommand\footnotetextcopyrightpermission[1]{} 
\makeatletter
\renewcommand\@formatdoi[1]{\ignorespaces}
\makeatother

\usepackage{natbib}
\usepackage{algorithm}
\usepackage[noend]{algorithmic}
\usepackage{amsmath}
\usepackage{geometry}
\usepackage{verbatim}
\usepackage{tabularx}
\usepackage{graphicx}
\usepackage{amsfonts}
\usepackage{cleveref}
\usepackage{multirow}
\usepackage{microtype}
\usepackage{enumitem}
\usepackage{hyperref}
\hypersetup{
    colorlinks=true,
    linkcolor=blue,
    filecolor=magenta,      
    urlcolor=cyan,
}
\newtheorem*{theorem*}{Theorem}
\newtheorem{theorem}{Theorem}[section]

\newtheorem{lemma}[theorem]{Lemma}

\newtheorem{definition}{Definition}[section]

\DeclareMathOperator*{\argmax}{arg\,max}

\newcommand{\floor}[1]{\lfloor #1 \rfloor}

\newcommand{\cP}{\mathcal{P}}

\newcommand{\cH}{{\mathcal{H}}}

\newcommand{\dt}{\mathrm{d}t}
\newcommand{\dx}{\mathrm{d}x}

\newcommand{\ddp}{\frac{\mathrm{d}}{\mathrm{d}p}}

\def\b1{{\bf 1}}

\DeclareMathOperator\sech{sech}

\title{An Elo-like System for Massive Multiplayer Competitions}

\author{Aram Ebtekar}
\affiliation{%
  \city{Vancouver}
  \state{BC}
  \country{Canada}
}
\email{aramebtech@gmail.com}

\author{Paul Liu}
\affiliation{%
  \institution{Stanford University}
  \city{Stanford}
  \state{CA}
  \country{USA}
}
\email{paul.liu@stanford.edu}

\date{}

\begin{document}

\begin{abstract}
Rating systems play an important role in competitive sports and games. They provide a measure of player skill, which incentivizes competitive performances and enables balanced match-ups. In this paper, we present a novel Bayesian rating system for contests with many participants. It is widely applicable to competition formats with discrete ranked matches, such as online programming competitions, obstacle courses races, and some video games. The simplicity of our system allows us to prove theoretical bounds on robustness and runtime. In addition, we show that the system \emph{aligns incentives}: that is, a player who seeks to maximize their rating will never want to underperform. Experimentally, the rating system rivals or surpasses existing systems in prediction accuracy, and computes faster than existing systems by up to an order of magnitude.
\looseness=-1
\end{abstract}

\maketitle
\pagestyle{plain}

\section{Introduction}

Competitions, in the form of sports, games, and examinations, have been with us since antiquity. Many competitions grade performances along a numerical scale, such as a score on a test or a completion time in a race. In the case of a college admissions exam or a track race, scores are standardized so that a given score on two different occasions carries the same meaning. However, in events that feature novelty, subjectivity, or close interaction, standardization is difficult. The Spartan Races, completed by millions of runners, feature a variety of obstacles placed on hiking trails around the world~\cite{Spartan}. Rock climbing, a sport to be added to the 2020 Olympics, likewise has routes set specifically for each competition. DanceSport, gymnastics, and figure skating competitions have a panel of judges who rank contestants against one another; these subjective scores are known to be noisy~\cite{DanceSport}. Most board games feature considerable inter-player interaction.
In all these cases, scores can only be used to compare and rank participants at the same event. Players, spectators, and contest organizers who are interested in comparing players' skill levels across different competitions will need to aggregate the entire history of such rankings. A strong player, then, is one who consistently wins against weaker players. To quantify skill, we need a \emph{rating system}.
\looseness=-1

Good rating systems are difficult to create, as they must balance several mutually constraining objectives. First and foremost, the rating system must be accurate, in that ratings provide useful predictors of contest outcomes. Second, the ratings must be efficient to compute: in video game applications, rating systems are predominantly used for matchmaking in massively multiplayer online games (such as Halo, CounterStrike, League of Legends, etc.)~\cite{HMG06, MCZ18, Y14}. These games have hundreds of millions of players playing tens of millions of games per day, necessitating certain latency and memory requirements for the rating system~\cite{AL09}. Third, the rating system must align incentives. That is, players should not modify their performance to ``game'' the rating system. Rating systems that can be gamed often create disastrous consequences to player-base, more often than not leading to the loss of players from the game~\cite{pokemongo}. Finally, the ratings provided by the system must be human-interpretable: ratings are typically represented to players as a single number encapsulating their overall skill, and many players want to understand and predict how their performances affect their rating~\cite{G95}.


Classically, rating systems were designed for two-player games. The famous Elo system~\cite{E61}, as well as its Bayesian successors Glicko and Glicko-2, have been widely applied to games such as Chess and Go~\cite{G95, G99, G12}. Both Glicko versions model each player's skill as a real random variable that evolves with time according to Brownian motion. Inference is done by entering these variables into the Bradley-Terry model, which predicts probabilities of game outcomes. Glicko-2 refines the Glicko system by adding a rating volatility parameter. Unfortunately, Glicko-2 is known to be flawed in practice, potentially incentivising players to lose. This was most notably exploited in the popular game of Pokemon Go~\cite{pokemongo}; see \Cref{sec:mono} for a discussion of this issue.

The family of Elo-like methods just described only utilize the binary outcome of a match. In settings where a scoring system provides a more fine-grained measure of match performance, Kovalchik~\cite{K20} has shown variants of Elo that are able to take advantage of score information. For competitions consisting of several set tasks, such as academic olympiads, Fori{\v{s}}ek~\cite{forivsektheoretical} developed a model in which each task gives a different ``response'' to the player: the total response then predicts match outcomes. However, such systems are often highly application-dependent and hard to calibrate.
\looseness=-1

Though Elo-like systems are widely used in two-player contests, one needn't look far to find competitions that involve much more than two players. Aside from the aforementioned sporting examples, there are video games such as CounterStrike and Halo, as well as academic olympiads: notably, programming contest platforms such as Codeforces, TopCoder, and Kaggle~\cite{Codeforces, TopCoder, Kaggle}. In these applications, the number of contestants can easily reach into the thousands. Some more recent works present interesting methods to deal with competitions between two teams \cite{HLW06, CJ16, LCFHH18, GFYLWTFC20}, but they do not present efficient extensions for settings in which players are sorted into more than two, let alone thousands, of distinct places.

In a many-player ranked competition, it is important to note that the pairwise match outcomes are not independent, as they would be in a series of 1v1 matches. Thus, TrueSkill~\cite{HMG06} and its variants~\cite{NS10, DHMG07, MCZ18} model a player's performance during each contest as a single random variable. The overall rankings are assumed to reveal the total order among these hidden performance variables, with various strategies used to model ties and teams. These TrueSkill algorithms are efficient in practice, successfully rating userbases that number well into the millions (the Halo series, for example, has over 60 million sales since 2001~\cite{Halo}).

The main disadvantage of TrueSkill is its complexity: originally developed by Microsoft for the popular Halo video game, TrueSkill performs approximate belief propagation on a factor graph, which is iterated until convergence~\cite{HMG06}. Aside from being less human-interpretable, this complexity means that, to our knowledge, there are no proofs of key properties such as runtime and incentive alignment. Even when these properties are discussed~\cite{MCZ18}, no rigorous justification is provided. In addition, we are not aware of any work that extends TrueSkill to non-Gaussian performance models, which might be desirable to limit the influence of outlier performances (see \Cref{sec:robust}).

It might be for these reasons that platforms such as Codeforces and TopCoder opted for their own custom rating systems. These systems are not published in academia and do not come with Bayesian justifications. However, they retain the formulaic simplicity of Elo and Glicko, extending them to settings with much more than two players. The Codeforces system includes ad hoc heuristics to distinguish top players, while curbing rampant inflation. TopCoder's formulas are more principled from a statistical perspective; however, it has a volatility parameter similar to Glicko-2, and hence suffers from similar exploits~\cite{forivsektheoretical}. Despite their flaws, these systems have been in place for over a decade, and have more recently gained adoption by additional platforms such as CodeChef and LeetCode~\cite{LeetCode, CodeChef}.

\paragraph{Our contributions} 
In this paper, we describe the Elo-MMR rating system, obtained by a principled approximation of a Bayesian model very similar to TrueSkill. It is fast, embarrassingly parallel, and makes accurate predictions. Most interesting of all, its simplicity allows us to rigorously analyze its properties: the ``MMR'' in the name stands for ``Massive'', ``Monotonic'', and ``Robust''. ``Massive'' means that it supports any number of players with a runtime that scales linearly; ``monotonic'' means that it \emph{aligns incentives} so that a rating-maximizing player always wants to perform well; ``robust'' means that rating changes are bounded, with the bound being smaller for more consistent players than for volatile players. Robustness turns out to be a natural byproduct of accurately modeling performances with heavy-tailed distributions, such as the logistic. TrueSkill is believed to satisfy the first two properties, albeit without proof, but fails robustness. Codeforces only satisfies aligned incentives, and TopCoder only satisfies robustness.

Experimentally, we show that Elo-MMR achieves state-of-the-art performance in terms of both prediction accuracy and runtime. In particular, we process the entire Codeforces database of over 300K rated users and 1000 contests in well under a minute, beating the existing Codeforces system by an order of magnitude while improving upon its accuracy. A difficulty we faced was the scarcity of efficient open-source rating system implementations. In an effort to aid researchers and practitioners alike, we provide open-source implementations of all rating systems, datasets, and additional processing used in our experiments at \url{https://github.com/EbTech/EloR/}.

\paragraph{Organization}
In \Cref{sec:bayes_model}, we formalize the details of our Bayesian model. We then show how to estimate player skill under this model in \Cref{sec:main-alg}, and develop some intuitions of the resulting formulas. As a further refinement, \Cref{sec:skill-drift} models skill evolutions from players training or atrophying between competitions. This modeling is quite tricky as we choose to retain players' momentum while ensuring it cannot be exploited for incentive-misaligned rating gains. \Cref{sec:properties} proves that the system as a whole satisfies several salient properties, the most critical of which is aligned incentives. Finally, we present experimental evaluations in \Cref{sec:experiments}.

\section{A Bayesian Model for Massive Competitions}
\label{sec:bayes_model}

We now describe the setting formally, denoting random variables by capital letters. A series of competitive \textbf{rounds}, indexed by $t=1,2,3,\ldots$, take place sequentially in time. Each round has a set of participating \textbf{players} $\cP_t$, which may in general overlap between rounds. A player's \textbf{skill} is likely to change with time, so we represent the skill of player $i$ at time $t$ by a real random variable $S_{i,t}$.

In round $t$, each player $i\in \cP_t$ competes at some \textbf{performance} level $P_{i,t}$, typically close to their current skill $S_{i,t}$. The deviations $\{P_{i,t}-S_{i,t}\}_{i\in\cP_t}$ are assumed to be i.i.d. and independent of $\{S_{i,t}\}_{i\in\cP_t}$.

Performances are not observed directly; instead, a ranking gives the relative order among all performances $\{P_{i,t}\}_{i\in\cP_t}$. In particular, ties are modelled to occur when performances are exactly equal, a zero-probability event when their distributions are continuous.\footnote{
The relevant limiting procedure is to treat performances within $\epsilon$-width buckets as ties, and letting $\epsilon\rightarrow 0$. This technicality appears in the proof of \Cref{thm:uniq-max}.} This ranking constitutes the observational \textbf{evidence} $E_t$ for our Bayesian updates. The rating system seeks to estimate the skill $S_{i,t}$ of every player at the present time $t$, given the historical round rankings $E_{\le t} := \{ E_1,\ldots,E_t \}$.

We overload the notation $\Pr$ for both probabilities and probability densities: the latter interpretation applies to zero-probability events, such as in $\Pr(S_{i,t} = s)$. We also use colons as shorthand for collections of variables differing only in a subscript: for instance, $P_{:,t}:=\{P_{i,t}\}_{i\in\cP_t}$. The joint distribution described by our Bayesian model factorizes as follows:
\begin{align}
    &\Pr(S_{:,:}, P_{:,:}, E_:) \label{eq:model}
    \\&= \prod_i \Pr(S_{i,0})
    \prod_{i,t} \Pr(S_{i,t}\mid S_{i,t-1})
    \prod_{i,t} \Pr(P_{i,t}\mid S_{i,t})
    \prod_t \Pr(E_t\mid P_{:,t}), \nonumber
\end{align}
\vspace{-1.5em}
\begin{align*}
    \text{where } \Pr(S_{i,0}) &\text{ is the initial skill prior,}
    \\\Pr(S_{i,t}\mid S_{i,t-1}) &\text{ is the skill evolution model (\Cref{sec:skill-drift}),}
    \\\Pr(P_{i,t}\mid S_{i,t}) &\text{ is the performance model, and}
    \\\Pr(E_t\mid P_{:,t}) &\text{ is the evidence model.}
\end{align*}
For the first three factors, we will specify log-concave distributions (see \Cref{def:log-concave}). The evidence model, on the other hand, is a deterministic indicator. It equals one when $E_t$ is consistent with the relative ordering among $\{P_{i,t}\}_{i\in\cP_t}$, and zero otherwise.

Finally, our model assumes that the number of participants $|\cP_t|$ is large. 
The main idea behind our algorithm is that, in sufficiently massive competitions, from the evidence $E_t$ we can infer very precise estimates for $\{P_{i,t}\}_{i\in\cP_t}$. Hence, we can treat these performances as if they were observed directly.

That is, suppose we have the skill prior at round $t$:
\begin{equation}
\label{eq:pi-s}
\pi_{i,t}(s) := \Pr(S_{i,t} = s \mid P_{i,<t}).
\end{equation}

Now, we observe $E_t$. By \Cref{eq:model}, it is conditionally independent of $S_{i,t}$, given $P_{i,\le t}$. By the law of total probability,
\begin{align*}
&\Pr(S_{i,t}=s \mid P_{i,<t},\,E_t)
\\&= \int \Pr(S_{i,t}=s \mid P_{i,<t},\,P_{i,t}=p) \Pr(P_{i,t}=p \mid P_{i,<t},\,E_t) \, \mathrm{d}p
\\&\rightarrow \Pr(S_{i,t}=s \mid P_{i,\le t}) \quad\text{almost surely as }|\mathcal P_t|\rightarrow\infty.
\end{align*}
The integral is intractable in general, since the performance posterior $\Pr(P_{i,t}=p \mid P_{i,<t},\,E_t)$ depends not only on player $i$, but also on our belief regarding the skills of all $j\in\cP_t$. However, in the limit of infinite participants, Doob's consistency theorem \cite{F63} implies that it concentrates at the true value $P_{i,t}$. Since our posteriors are continuous, the convergence holds for all $s$ simultaneously.

Indeed, we don't even need the full evidence $E_t$. Let $E^L_{i,t} = \{j\in\cP:P_{j,t}>P_{i,t}\}$ be the set of players against whom $i$ lost, and $E^W_{i,t} = \{j\in\cP:P_{j,t}<P_{i,t}\}$ be the set of players against whom $i$ won. That is, we only see who wins, draws, and loses against $i$. $P_{i,t}$ remains identifiable using only $(E^L_{i,t}, E^W_{i,t})$, which will be more convenient for our purposes.

Passing to the limit $|\cP_t|\rightarrow\infty$ serves to justify some common simplifications made by algorithms such as TrueSkill: since conditioning on $P_{i,\le t}$ makes the skills of different players independent of one another, it becomes accurate to model them as such. In addition to simplifying derivations, this fact ensures that a player's posterior is unaffected by rounds in which they are not a participant, arguably a desirable property in its own right. Furthermore, $P_{i,\le t}$ being a sufficient statistic for skill prediction renders any additional information, such as domain-specific raw scores, redundant.

Finally, a word on the rate of convergence. Suppose we want our estimate to be within $\epsilon$ of $P_{i,t}$, with probability at least $1-\delta$. By asymptotic normality of the posterior~\cite{F63}, it suffices to have $O(\frac 1{\epsilon^2}\sqrt{\log \frac 1\delta})$ participants.

When the initial prior, performance model, and evolution model are all Gaussian, treating $P_{i,t}$ as certain is the \emph{only} simplifying approximation we will make; that is, in the limit $|\cP_t|\rightarrow\infty$, our method performs \emph{exact} inference on \Cref{eq:model}. In the following sections, we focus some attention on generalizing the performance model to non-Gaussian log-concave families, parametrized by location and scale. We will use the logistic distribution as a running example and see that it induces robustness; however, our framework is agnostic to the specific distributions used.

The \textbf{rating} $\mu_{i,t}$ of player $i$ after round $t$ should be a statistic that summarizes their posterior distribution: we'll use the maximum a posteriori (MAP) estimate, obtained by setting $s$ to maximize the posterior $\Pr(S_{i,t}=s \mid P_{i,\le t})$. By Bayes' rule,
\begin{equation}
\label{eq:new-obj}
\mu_{i,t} := \argmax_{s} \pi_{i,t}(s) \Pr(P_{i,t} \mid S_{i,t}=s).
\end{equation}
This objective suggests a two-phase algorithm to update each player $i\in\cP_t$ at round $t$. In phase one, we estimate $P_{i,t}$ from $(E^L_{i,t}, E^W_{i,t})$. By Doob's consistency theorem, our estimate is extremely precise when $|\cP_t|$ is large, so we assume it to be exact. In phase two, we update our posterior for $S_{i,t}$ and the rating $\mu_{i,t}$ according to \Cref{eq:new-obj}.

We will occasionally make use of the \textbf{prior rating}, defined as
\begin{equation*}
\mu_{i,t}^\pi := \argmax_{s} \pi_{i,t}(s).
\end{equation*}
\section{A two-phase algorithm for skill estimation}
\label{sec:main-alg}
    \subsection{Performance estimation}
\label{sec:performance}

In this section, we describe the first phase of Elo-MMR. For notational convenience, we assume all probability expressions to be conditioned on the \textbf{prior context} $P_{i,< t}$, and omit the subscript $t$.

Our prior belief on each player's skill $S_i$ implies a prior distribution on $P_i$. Let's denote its probability density function (pdf) by
\looseness=-1
\begin{equation}
\label{eq:perf-prior} 
f_i(p) := \Pr(P_i = p) = \int \pi_i(s) \Pr(P_i = p \mid S_i=s) \,\mathrm{d}s,
\end{equation}
where $\pi_i(s)$ was defined in \Cref{eq:pi-s}. Let
\[F_i(p) := \Pr(P_i\le p) = \int_{-\infty}^p f_i(x) \,\dx,\]
be the corresponding cumulative distribution function (cdf). For the purpose of analysis, we'll also define the following ``loss'', ``draw'', and ``victory'' functions:
\begin{align*}
l_i(p) &:= \ddp\ln(1-F_i(p)) = \frac{-f_i(p)}{1 - F_i(p)},
\\d_i(p) &:= \ddp\ln f_i(p) = \frac{f'_i(p)}{f_i(p)},
\\v_i(p) &:= \ddp\ln F_i(p) = \frac{f_i(p)}{F_i(p)}.
\end{align*}

Evidently, $l_i(p) < 0 < v_i(p)$. Now we define what it means for the deviation $P_i - S_i$ to be log-concave.
\begin{definition}
\label{def:log-concave}
An absolutely continuous random variable on a convex domain is \textbf{log-concave} if its probability density function $f$ is positive on its domain and satisfies
\[f(\theta x + (1-\theta) y) > f(x)^\theta f(y)^{1-\theta},\;\forall\theta\in(0,1),x\neq y.\]
\end{definition}

We note that log-concave distributions appear widely, and include the Gaussian and logistic distributions used in Glicko, TrueSkill, and many others. We'll see inductively that our prior $\pi_i$ is log-concave at every round. Since log-concave densities are closed under convolution~\cite{concave}, the independent sum $P_i=S_i+(P_i-S_i)$ is also log-concave. The following lemma (proved in the appendix) makes log-concavity very convenient:
\begin{lemma}
\label{lem:decrease}
If $f_i$ is continuously differentiable and log-concave, then the functions $l_i,d_i,v_i$ are continuous, strictly decreasing, and
\[l_i(p) < d_i(p) < v_i(p) \text{ for all }p.\]
\end{lemma}

For the remainder of this section, we fix the analysis with respect to some player $i$. As argued in \Cref{sec:bayes_model}, $P_i$ concentrates very narrowly in the posterior. Hence, we can estimate $P_i$ by its MAP, choosing $p$ so as to maximize:
\[\Pr(P_i=p\mid E^L_i,E^W_i) \propto f_i(p) \Pr(E^L_i,E^W_i\mid P_i=p).\]

Define $j\succ i$, $j\prec i$, $j\sim i$ as shorthand for $j\in E^L_i$, $j\in E^W_i$, $j\in \mathcal P\setminus (E^L_i\cup E^W_i)$ (that is, $P_j>P_i$, $P_j<P_i$, $P_j=P_i$), respectively. The following theorem yields our MAP estimate:
\begin{theorem}
\label{thm:uniq-max}
Suppose that for all $j$, $f_j$ is continuously differentiable and log-concave. Then the unique maximizer of $\Pr(P_i=p\mid E^L_i,E^W_i)$ is given by the unique zero of
\[Q_i(p) := \sum_{j \succ i} l_j(p) + \sum_{j \sim i} d_j(p) + \sum_{j \prec i} v_j(p).\]
\end{theorem}
The proof is relegated to the appendix. Intuitively, we're saying that the performance is the balance point between appropriately weighted wins, draws, and losses. Let's look at two specializations of our general model, to serve as running examples in this paper.

\paragraph{Gaussian performance model}
If both $S_j$ and $P_j-S_j$ are assumed to be Gaussian with known means and variances, then their independent sum $P_j$ will also be a known Gaussian. It is analytic and log-concave, so \Cref{thm:uniq-max} applies.

We substitute the well-known Gaussian pdf and cdf for $f_j$ and $F_j$, respectively. A simple binary search, or faster numerical techniques such as the Illinois algorithm or Newton's method, can be employed to solve for the maximizing $p$.

\begin{figure}
    \centering
    \includegraphics[width=1.05\columnwidth]{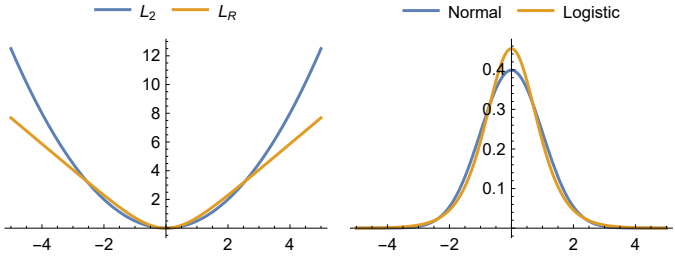}
    \caption{$L_2$ versus $L_R$ for typical values (left). Gaussian versus logistic probability density functions (right).}
    \label{fig:l2-lr-plot}
\end{figure}

\paragraph{Logistic performance model}
Now we assume the performance deviation $P_j-S_j$ has a logistic distribution with mean 0 and variance $\beta^2$. In general, the rating system administrator is free to set $\beta$ differently for each contest. Since shorter contests tend to be more variable, one reasonable choice might be to make $1/\beta^2$ proportional to the contest duration.

Given the mean and variance of the skill prior, the independent sum $P_j = S_j + (P_j-S_j)$ would have the same mean, and a variance that's increased by $\beta^2$. Unfortunately, we'll see that the logistic performance model implies a form of skill prior from which it's tough to extract a mean and variance. Even if we could, the sum does not yield a simple distribution.

For experienced players, we expect $S_j$ to contribute much less variance than $P_j-S_j$; thus, in our heuristic approximation, we take $P_j$ to have the same form of distribution as the latter. That is, we take $P_j$ to be logistic, centered at the prior rating $\mu^\pi_j = \argmax \pi_j$, with variance $\delta_j^2 = \sigma_j^2 + \beta^2$, where $\sigma_j$ will be given by \Cref{eq:variance}. This distribution is analytic and log-concave, so the same methods based on \Cref{thm:uniq-max} apply. 
Define the scale parameter $\bar\delta_j := \frac{\sqrt{3}}{\pi} \delta_j$. A logistic distribution with variance $\delta_j^2$ has cdf and pdf:
\begin{align*}
F_j(x) &= \frac { 1 } { 1 + e^{-(x-\mu^\pi_j)/\bar\delta_j} }
= \frac 12 \left(1 + \tanh\frac{x-\mu^\pi_j}{2\bar\delta_j} \right),
\\f_j(x) &= \frac { e^{(x-\mu^\pi_j)/\bar\delta_j} } { \bar\delta_j\left( 1 + e^{(x-\mu^\pi_j)/\bar\delta_j} \right)^2}
= \frac { 1 } { 4\bar\delta_j} \sech^2\frac{x-\mu^\pi_j}{2\bar\delta_j}.
\end{align*}

The logistic distribution satisfies two very convenient relations:
\begin{align*}
F'_j(x) = f_j(x) &= F_j(x) (1 - F_j(x)) / \bar\delta_j,
\\f'_j(x) &= f_j(x) (1 - 2F_j(x)) / \bar\delta_j,
\end{align*}
from which it follows that
\[d_j(p)
= \frac{1 - 2F_j(p)}{\bar\delta}
= \frac{-F_j(p)}{\bar\delta} + \frac{1 - F_j(p)}{\bar\delta}
= l_j(p) + v_j(p).\]

In other words, a tie counts as the sum of a win and a loss. This can be compared to the approach (used in Elo, Glicko, TopCoder, and Codeforces) of treating each tie as half a win plus half a loss.\footnote{Elo-MMR, too, can be modified to split ties into half win plus half loss. It's easy to check that \Cref{lem:decrease} still holds if $d_j(p)$ is replaced by
$w_l l_j(p) + w_v v_j(p)$
for some $w_l,w_v\in [0,1]$ with $|w_l-w_v|<1$.
In particular, we can set $w_l=w_v=0.5$. The results in \Cref{sec:properties} won't be altered by this change.}

Finally, putting everything together:
\[Q_i(p) = \sum_{j \succeq i} l_j(p) + \sum_{j \preceq i} v_j(p)
= \sum_{j \succeq i} \frac{-F_j(p)}{\bar\delta_j} + \sum_{j \preceq i} \frac{1 - F_j(p)}{\bar\delta_j}.\]
Our estimate for $P_i$ is the zero of this expression. The terms on the right correspond to probabilities of winning and losing against each player $j$, weighted by $1/\bar\delta_j$. Accordingly, we can interpret $\sum_{j\in \cP} (1-F_j(p))/\bar\delta_j$ as a weighted expected rank of a player whose performance is $p$. Similar to the performance computations in Codeforces and TopCoder, $P_i$ can thus be viewed as the performance level at which one's expected rank would equal $i$'s actual rank.

    \subsection{Belief update}
\label{sec:belief}

Having estimated $P_{i,t}$ in the first phase, the second phase is rather simple. Ignoring normalizing constants, \Cref{eq:new-obj} tells us that the pdf of the skill posterior can be obtained as the pointwise product of the pdfs of the skill prior and the performance model. When both factors are differentiable and log-concave, then so is their product. Its maximum is the new rating $\mu_{i,t}$; let's see how to compute it for the same two specializations of our model.

\paragraph{Gaussian skill prior and performance model}
When the skill prior and performance model are Gaussian with known means and variances, multiplying their pdfs yields another known Gaussian. Hence, the posterior is compactly represented by its mean $\mu_{i,t}$, which coincides with the MAP and rating; and its variance $\sigma_{i,t}^2$, which is our \textbf{uncertainty} regarding the player's skill.

\paragraph{Logistic performance model}
When the performance model is non-Gaussian, the multiplication does not simplify so easily. By \Cref{eq:new-obj}, each round contributes an additional factor to the belief distribution. In general, we allow it to consist of a collection of simple log-concave factors, one for each round in which player $i$ has participated. Denote the participation history by
\[\cH_{i,t} := \{k\in\{1,\ldots,t\}:i\in\mathcal P_k\}.\]

Since each player can be considered in isolation, we'll omit the subscript $i$. Specializing to the logistic setting, each $k\in\cH_t$ contributes a logistic factor to the posterior, with mean $p_k$ and variance $\beta_k^2$. We still use a Gaussian initial prior, with mean and variance denoted by $p_0$ and $\beta_0^2$, respectively. Postponing the discussion of skill evolution to \Cref{sec:skill-drift}, for the moment we assume that $S_k=S_0$ for all $k$. The posterior pdf, up to normalization, is then
\begin{align}
&\pi_0(s) \prod_{k\in\cH_t} \Pr(P_k=p_k \mid S_k=s) \nonumber
\\&\propto \exp\left( -\frac{(s-p_0)^2}{2\beta_0^2} \right) \label{eq:posterior}
\prod_{k\in\cH_t} \sech^{2}\left( \frac\pi{\sqrt{12}} \frac{s-p_k} {\beta_k} \right).
\end{align}

Maximizing the posterior density amounts to minimizing its negative logarithm. Up to a constant offset, this is given by
\begin{align*}
L(s) &:= L_2\left(\frac{s-p_0}{\beta_0}\right)
+ \sum_{k\in\cH_t} L_R\left(\frac{s-p_k}{\beta_k}\right),
\\\text{where }L_2(x) &:= \frac 12 x^2\text{ and }
L_R(x) := 2\ln\left(\cosh \frac{\pi x}{\sqrt{12}}\right).
\end{align*}
\begin{equation}
\label{eq:loss}
\text{Thus, }L'(s) = \frac{s-p_0}{\beta_0^2} + \sum_{k\in\cH_t} \frac{\pi}{\beta_k\sqrt{3}} \tanh \frac{(s-p_k)\pi}{\beta_k\sqrt{12}}.
\end{equation}

$L'$ is continuous and strictly increasing in $s$, so its zero is unique: it is the MAP $\mu_t$. Similar to what we did in the first phase, we can solve for $\mu_t$ with either binary search or Newton's method.

We pause to make an important observation. From \Cref{eq:loss}, the rating carries a rather intuitive interpretation: Gaussian factors in $L$ become $L_2$ penalty terms, whereas logistic factors take on a more interesting form as $L_R$ terms. From \Cref{fig:l2-lr-plot}, we see that the $L_R$ term behaves quadratically near the origin, but linearly at the extremities, effectively interpolating between $L_2$ and $L_1$ over a scale of magnitude $\beta_k$ 

It is well-known that minimizing a sum of $L_2$ terms pushes the argument towards a weighted mean, while minimizing a sum of $L_1$ terms pushes the argument towards a weighted median. With $L_R$ terms, the net effect is that $\mu_t$ acts like a robust average of the historical performances $p_k$. Specifically, one can check that
\[\mu_t = \frac{\sum_k w_k p_k}{\sum_k w_k}, \text{ where } w_0 := \frac{1}{\beta_0^2} \text{ and }\]
\begin{equation}
\label{eq:average}
w_k := \frac{\pi}{(\mu_t-p_k)\beta_k\sqrt{3}}\tanh\frac{(\mu_t-p_k)\pi}{\beta_k\sqrt{12}} \text{ for }k\in\cH_t.
\end{equation}

$w_k$ is close to $1/\beta_k^2$ for typical performances, but can be up to $\pi^2/6$ times more as $|\mu_t-p_k| \rightarrow 0$, or vanish as $|\mu_t-p_k| \rightarrow\infty$. This feature is due to the thicker tails of the logistic distribution, as compared to the Gaussian, resulting in an algorithm that resists drastic rating changes in the presence of a few unusually good or bad performances. We'll formally state this \emph{robustness} property in \Cref{thm:robust}.


\paragraph{Estimating skill uncertainty} While there is no easy way to compute the variance of a posterior in the form of \Cref{eq:posterior}, it will be useful to have some estimate $\sigma_t^2$ of uncertainty. There is a simple formula in the case where all factors are Gaussian. Since moment-matched logistic and normal distributions are relatively close (c.f. \Cref{fig:l2-lr-plot}), we apply the same formula:
\begin{equation}
\label{eq:variance}
\frac{1}{\sigma_t^2} := \sum_{k\in\{0\}\cup\cH_t}\frac{1}{\beta_k^2}.
\end{equation}



\section{Skill evolution over time}
\label{sec:skill-drift}

Factors such as training and resting will change a player's skill over time. If we model skill as a static variable, our system will eventually grow so confident in its estimate that it will refuse to admit substantial changes. To remedy this, we introduce a skill evolution model, so that in general $S_t \neq S_{t'}$ for $t \neq t'$. Now rather than simply being equal to the previous round's posterior, the skill prior at round $t$ is given by
\begin{equation}
\label{eq:drift}
\pi_t(s) = \int \Pr(S_t = s \mid S_{t-1} = x) \Pr(S_{t-1} = x \mid P_{<t}) \,\dx.
\end{equation}

The factors in the integrand are the skill evolution model and the previous round's posterior, respectively. Following other Bayesian rating systems (e.g., Glicko, Glicko-2, and TrueSkill~\cite{G99, G12, HMG06}), we model the skill diffusions $S_t-S_{t-1}$ as independent zero-mean Gaussians. That is, $\Pr(S_t \mid S_{t-1}=x)$ is a Gaussian with mean $x$ and some variance $\gamma_t^2$. The Glicko system sets $\gamma_t^2$ proportionally to the time elapsed since the last update, corresponding to a continuous Brownian motion. Codeforces and TopCoder simply set $\gamma_t$ to a constant when a player participates, and zero otherwise, corresponding to changes that are in proportion to how often the player competes. Now we are ready to complete the two specializations of our rating system.

\paragraph{Gaussian skill prior and performance model}
If both the prior and performance distributions at round $t-1$ are Gaussian, then the posterior is also Gaussian. Adding an independent Gaussian diffusion to our posterior on $S_{t-1}$ yields a Gaussian prior on $S_t$. By induction, the skill belief distribution forever remains Gaussian. This Gaussian specialization of the Elo-MMR framework lacks the R for robustness (see \Cref{thm:robust}), so we call it Elo-MM$\chi$.

\paragraph{Logistic performance model}
After a player's first contest round, the posterior in \Cref{eq:posterior} becomes non-Gaussian, rendering the integral in \Cref{eq:drift} intractable.

A very simple approach would be to replace the full posterior in \Cref{eq:posterior} by a Gaussian approximation with mean $\mu_t$ (equal to the posterior MAP) and variance $\sigma_t^2$ (given by \Cref{eq:variance}). As in the previous case, applying diffusions in the Gaussian setting is a simple matter of adding means and variances.

With this approximation, no memory is kept of the individual performances $P_t$. Priors are simply Gaussian, while posterior densities are the product of two factors: the Gaussian prior, and a logistic factor corresponding to the latest performance. To ensure robustness (see \Cref{sec:robust}), $\mu_t$ is computed as the argmax of this posterior \emph{before} replacement by its Gaussian approximation. We call the rating system that takes this approach Elo-MMR($\infty$).

As the name implies, it turns out to be a special case of Elo-MMR($\rho$). In the general setting with $\rho \in [0,\infty)$, we keep the full posterior from \Cref{eq:posterior}. Since we cannot tractably compute the effect of a Gaussian diffusion, we seek a heuristic derivation of the next round's prior, retaining a form similar to \Cref{eq:posterior} while satisfying many of the same properties as the intended diffusion.

\subsection{Desirable properties of a ``pseudodiffusion''}
\label{sec:desirable-props}
We begin by listing some properties that our skill evolution algorithm, henceforth called a ``pseudodiffusion'', should satisfy. The first two properties are natural:
\begin{itemize}[leftmargin=*]
\item \emph{Aligned incentives.} First and foremost, the pseudodiffusion must not break the aligned incentives property of our rating system. That is, a rating-maximizing player should never be motivated to lose on purpose (\Cref{thm:mono}).
\item \emph{Rating preservation.} The pseudodiffusion must not alter the $\argmax$ of the belief density. That is, the rating of a player should not change: $\mu^\pi_t = \mu_{t-1}$.
\end{itemize}
In addition, we borrow four properties of Gaussian diffusions:
\begin{itemize}[leftmargin=*]
\item \emph{Correct magnitude.} Pseudodiffusion with parameter $\gamma^2$ must increase the skill uncertainty, as measured by \Cref{eq:variance}, by $\gamma^2$.
\item \emph{Composability.} Two pseudodiffusions applied in sequence, first with parameter $\gamma_1^2$ and then with $\gamma_2^2$, must have the same effect as a single pseudodiffusion with parameter $\gamma_1^2 + \gamma_2^2$.
\item \emph{Zero diffusion.} In the limit as $\gamma \rightarrow 0$, the effect of pseudodiffusion must vanish, i.e., not alter the belief distribution.
\item \emph{Zero uncertainty.} In the limit as $\sigma_{t-1}\rightarrow 0$ (i.e., when the previous rating $\mu_{t-1}$ is a perfect estimate of $S_{t-1}$), our belief on $S_t$ must become Gaussian with mean $\mu_{t-1}$ and variance $\gamma^2$. Finer-grained information regarding the prior history $P_{\le t}$ must be erased.
\end{itemize}
In particular, Elo-MMR($\infty$) fails the \emph{zero diffusion} property because it simplifies the belief distribution, even when $\gamma=0$. In the proof of \Cref{thm:diffuse-prop}, we'll see that Elo-MMR($0$) fails the \emph{zero uncertainty} property. Thus, it is in fact necessary to have $\rho$ strictly positive and finite. In \Cref{sec:robust}, we'll come to interpret $\rho$ as a kind of inverse momentum.

\subsection{A heuristic pseudodiffusion algorithm}
\label{sec:pseudodiffusion}
Each factor in the posterior (see \Cref{eq:posterior}) has a parameter $\beta_k$. Define a factor's \textbf{weight} to be $w_k := 1/\beta_k^2$, which by \Cref{eq:variance} contributes to the \textbf{total weight} $\sum_k w_k=1/\sigma_t^2$. Here, unlike in \Cref{eq:average}, $w_k$ does not depend on $|\mu_t-p_k|$.

The approximation step of Elo-MMR($\infty$) replaces all the logistic factors by a single Gaussian whose variance is chosen to ensure that the total weight is preserved. In addition, its mean is chosen to preserve the player's rating, given by the unique zero of \Cref{eq:loss}. Finally, the diffusion step of Elo-MMR($\infty$) increases the Gaussian's variance, and hence the player's skill uncertainty, by $\gamma_t^2$; this corresponds to a decay in the weight.

To generalize the idea, we interleave the two steps in a continuous manner. The approximation step becomes a \textbf{transfer step}: rather than replace the logistic factors outright, we take away the same fraction from each of their weights, and place the sum of removed weights onto a new Gaussian factor. The diffusion step becomes a \textbf{decay step}, reducing each factor's weight by the same fraction, chosen such that the overall uncertainty is increased by $\gamma_t^2$.

To make the idea precise, we generalize the posterior from \Cref{eq:posterior} with fractional \textbf{multiplicities} $\omega_k$, initially set to $1$ for each $k\in\{0\}\cup\cH_t$. The $k$'th factor is raised to the power $\omega_k$; in \Cref{eq:loss,eq:variance}, the corresponding term is multiplied by $\omega_k$. In other words, the latter equation is replaced by
\[\frac{1}{\sigma_t^2} := \sum_{k\in\{0\}\cup\cH_t}w_k,\text{ where }w_k := \frac{\omega_k}{\beta_k^2}.\]

For $\rho\in [0,\infty]$, the Elo-MMR($\rho$) algorithm continuously and simultaneously performs transfer and decay, with transfer proceeding at $\rho$ times the rate of decay. Holding $\beta_k$ fixed, changes to $\omega_k$ can be described in terms of changes to $w_k$:
\begin{align*}
\dot w_0 &= -r(t)w_0 + \rho r(t) \sum_{k\in\cH_t} w_k,
\\\dot w_k &= -(1+\rho)r(t)w_k \quad\text{for }k\in\cH_t,
\end{align*}
where the arbitrary decay rate $r(t)$ can be eliminated by a change of variable $\mathrm{d}\tau = r(t)\dt$. After some time $\Delta\tau$, the total weight will have decayed by a factor $\kappa := e^{-\Delta\tau}$, resulting in the new weights:
\begin{align*}
w_0^{new} &= \kappa w_0 + \left(\kappa-\kappa^{1+\rho}\right)\sum_{k\in\cH_t} w_k,
\\w_k^{new} &= \kappa^{1+\rho}w_k \quad\text{for }k\in\cH_t.
\end{align*}
In order for the uncertainty to increase from $\sigma_{t-1}^2$ to $\sigma_{t-1}^2+\gamma_t^2$, we must solve $\kappa/\sigma_{t-1}^2 = 1/(\sigma_{t-1}^2+\gamma_t^2)$ for the decay factor:

\setlength{\floatsep}{0pt}
\setlength{\textfloatsep}{1em}
\begin{algorithm}[t]
\caption{Elo-MMR($\rho,\beta, \gamma$)}
\label{alg:main}
\begin{algorithmic}
\FORALL{rounds $t$}
\FORALL{players $i\in\mathcal P_t$ in parallel}
\IF{$i$ has never competed before}
\STATE {$\mu_i, \sigma_i \gets \mu_{newcomer}, \sigma_{newcomer}$}
\STATE {$p_i, w_i \gets [\mu_i], [1/\sigma_i^2]$}
\ENDIF
\STATE diffuse($i,\gamma,\rho$)
\STATE $\mu^\pi_i, \delta_i \gets \mu_i,\sqrt{\sigma_i^2 + \beta^2}$
\ENDFOR
\FORALL{$i\in\mathcal P_t$ in parallel}
\STATE update($i,E_t,\beta$)
\ENDFOR
\ENDFOR
\end{algorithmic}
\end{algorithm}
\begin{algorithm}[t]
\caption{diffuse($i,\gamma,\rho$)}
\label{alg:diffuse}
\begin{algorithmic}
\STATE $\kappa \gets (1+\gamma^2/\sigma_i^2)^{-1}$
\STATE $w_G, w_L \gets \kappa^\rho w_{i,0}, (1-\kappa^\rho) \sum_{k} w_{i,k}$
\STATE $p_{i,0} \gets (w_G p_{i,0} + w_L \mu_i) / (w_G+w_L)$
\STATE $w_{i,0} \gets \kappa (w_G+w_L)$
\FORALL{$k\ne 0$}
\STATE $w_{i,k} \gets \kappa^{1+\rho}w_{i,k}$
\ENDFOR
\STATE $\sigma_i \gets \sigma_i / \sqrt\kappa$
\end{algorithmic}
\end{algorithm}
\begin{algorithm}[t]
\caption{update($i,E,\beta$)}
\label{alg:update}
\begin{algorithmic}
\STATE $p \gets \mathrm{zero}\left(\sum_{j\preceq i}\frac{1}{\delta_j}\left( \tanh\frac {x - \mu^\pi_j} {2\bar\delta_j} - 1 \right) + \sum_{j\succeq i}\frac{1}{\delta_j}\left( \tanh\frac {x - \mu^\pi_j} {2\bar\delta_j} + 1 \right)\right)$
\STATE $p_i$.push($p$)
\STATE $w_i$.push($1/\beta^2$)
\STATE $\mu_i \gets \mathrm{zero}\left(w_{i,0}(x-p_{i,0}) + \sum_{k\ne 0} \frac{w_{i,k}\beta^2}{\bar\beta} \tanh \frac {x-p_{i,k}} {2\bar\beta}\right)$
\end{algorithmic}
\end{algorithm}

\[\kappa = \left(1 + \frac{\gamma_t^2}{\sigma_{t-1}^2}\right)^{-1}.\]
In order for this operation to preserve ratings, the transferred weight must be centered at $\mu_{t-1}$; see \Cref{alg:diffuse} for details.

\Cref{alg:main} details the full Elo-MMR($\rho$) rating system. The main loop runs whenever a round of competition takes place. First, new players are initialized with a Gaussian prior. Then, changes in player skill are modeled by \Cref{alg:diffuse}. Given the round rankings $E_t$, the first phase of \Cref{alg:update} solves an equation to estimate $P_t$. Finally, the second phase solves another equation for the rating $\mu_t$. 

The hyperparameters $\rho,\beta,\gamma$ are domain-dependent, and can be set by standard hyperparameter search techniques. For convenience, we assume $\beta$ and $\gamma$ are fixed and use the shorthand $\bar\beta_k := \frac{\sqrt{3}}{\pi} \beta_k$.

\begin{theorem}
\label{thm:diffuse-prop}
\Cref{alg:diffuse} with $\rho\in(0,\infty)$ meets all of the properties listed in \Cref{sec:desirable-props}.
\end{theorem}

\begin{proof}
We go through each of the six properties in order.
\begin{itemize}[leftmargin=*]
    \item \emph{Aligned incentives.} This property will be stated in \Cref{thm:mono}. To ensure that its proof carries through, the relevant facts to note here are that the pseudodiffusion algorithm ignores the performances $p_k$, and centers the transferred Gaussian weight at the rating $\mu_{t-1}$, which is trivially monotonic in $\mu_{t-1}$.
    \item \emph{Rating preservation.} Recall that the rating is the unique zero of $L'$, defined in \Cref{eq:loss}. To see that this zero is preserved, note that the decay and transfer operations multiply $L'$ by constants ($\kappa$ and $\kappa^\rho$, respectively), before adding the new Gaussian term, whose contribution to $L'$ is zero at its center.
    \item \emph{Correct magnitude.} Follows from our derivation for $\kappa$.
    \item \emph{Composability.} Follows from \emph{correct magnitude} and the fact that every pseudodiffusion follows the same differential equations.
    \item \emph{Zero diffusion.} As $\gamma\rightarrow 0$, $\kappa\rightarrow 1$. Provided that $\rho<\infty$, we also have $\kappa^\rho\rightarrow 1$. Hence, for all $k\in\{0\}\cup\cH_t$, $w_k^{new} \rightarrow w_k$.
    \item \emph{Zero uncertainty.} As $\sigma_{t-1}\rightarrow 0$, $\kappa\rightarrow 0$. The total weight decays from $1/\sigma_{t-1}^2$ to $\gamma^2$. Provided that $\rho > 0$, we also have $\kappa^\rho\rightarrow 0$, so these weights transfer in their entirety, leaving behind a Gaussian with mean $\mu_{t-1}$, variance $\gamma^2$, and no additional history. \qedhere
\end{itemize}
\end{proof}

\section{Theoretical Properties}
\label{sec:properties}
In this section, we see how the simplicity of the Elo-MMR formulas enables us to rigorously prove that the rating system aligns incentives, is robust, and is computationally efficient.

\vspace{-.6em}
\subsection{Aligned incentives}
\label{sec:mono}

To demonstrate the need for \emph{aligned incentives}, let's look at the consequences of violating this property in the TopCoder and Glicko-2 rating systems. These systems track a ``volatility'' for each player, which estimates the variance of their performances. A player whose recent performance history is more consistent would be assigned a lower volatility score, than one with wild swings in performance. The volatility acts as a multiplier on rating changes; thus, players with an extremely low or high performance will have their subsequent rating changes amplified.

While it may seem like a good idea to boost changes for players whose ratings are poor predictors of their performance, this feature has an exploit. By intentionally performing at a weaker level, a player can amplify future increases to an extent that more than compensates for the immediate hit to their rating. A player may even ``farm'' volatility by alternating between very strong and very weak performances. After acquiring a sufficiently high volatility score, the strategic player exerts their honest maximum performance over a series of contests. The amplification eventually results in a rating that exceeds what would have been obtained via honest play. This type of exploit was discovered in both TopCoder competitions and the Pokemon Go video game~\cite{forivsektheoretical, pokemongo}. For a detailed example, see Table 5.3 of~\cite{forivsektheoretical}.

Remarkably, Elo-MMR combines the best of both worlds: we'll see in \Cref{sec:robust} that, for $\rho\in (0,\infty)$, Elo-MMR($\rho$) also boosts changes to inconsistent players. And yet, as we'll prove in this section, no such strategic incentive exists in \emph{any} version of Elo-MMR.
\looseness=-1

Recall that, for the purposes of the algorithm, the performance $p_i$ is defined to be the unique zero of the function $Q_i(p) := \sum_{j \succ i} l_j(p) + \sum_{j \sim i} d_j(p) + \sum_{j \prec i} v_j(p)$, whose terms $l_i,d_i,v_i$ are contributed by opponents against whom $i$ lost, drew, or won, respectively. Wins (losses) are always positive (negative) contributions to a player's performance score:

\begin{lemma}
\label{lem:mono-term}
Adding a win term to $Q_i(\cdot)$, or replacing a tie term by a win term, always increases its zero. Conversely, adding a loss term, or replacing a tie term by a loss term, always decreases it.
\end{lemma}

\begin{proof}
By \Cref{lem:decrease}, $Q_i(p)$ is decreasing in $p$. Thus, adding a positive term will increase its zero whereas adding a negative term will decrease it. The desired conclusion follows by noting that, for all $j$ and $p$, $v_j(p)$ and $v_j(p)-d_j(p)$ are positive, whereas $l_j(p)$ and $l_j(p)-d_j(p)$ are negative.
\end{proof}

While not needed for our main result, a similar argument shows that performance scores are monotonic across the round standings:

\begin{theorem}
If $i \succ j$ (that is, player $i$ beats $j$) in a given round, then player $i$ and $j$'s performance estimates satisfy $p_i > p_j$.
\end{theorem}

\begin{proof}
If $i \succ j$ with $i,j$ adjacent in the rankings, then
\[Q_i(p) - Q_j(p) = \sum_{k\sim i}(d_k(p) - l_k(p)) + \sum_{k\sim j}(v_k(p) - d_k(p)) > 0.\]
for all $p$. Since $Q_i$ and $Q_j$ are decreasing functions, it follows that $p_i > p_j$. By induction, this result extends to the case where $i,j$ are not adjacent in the rankings.
\end{proof}

What matters for incentives is that performance scores be \emph{counterfactually} monotonic; meaning, if we were to alter the round standings, a strategic player will always prefer to place higher:
\begin{lemma}
\label{lem:mono-perf}
In any given round, holding fixed the relative ranking of all players other than $i$ (and holding fixed all preceding rounds), the performance $p_i$ is a monotonic function of player i's prior rating and of player $i$'s rank in this round.
\end{lemma}

\begin{proof}
Monotonicity in the rating follows directly from monotonicity of the self-tie term $d_i$ in $Q_i$. Since an upward shift in the rankings can only convert losses to ties to wins, monotonicity in contest rank follows from \Cref{lem:mono-term}.
\end{proof}

Having established the relationship between round rankings and performance scores, the next step is to prove that, even with hindsight, players will always prefer their performance scores to be as high as possible:

\begin{lemma}
\label{lem:mono-rate}
Holding fixed the set of contest rounds in which a player has participated, their current rating is monotonic in each of their past performance scores.
\end{lemma}

\begin{proof}
The player's rating is given by the zero of $L'$ in \Cref{eq:loss}. The pseudodiffusions of \Cref{sec:skill-drift} modify each of the $\beta_k$ in a manner that does not depend on any of the $p_k$, so they are fixed for our purposes. Hence, $L'$ is monotonically increasing in $s$ and decreasing in each of the $p_k$. Therefore, its zero is monotonically increasing in each of the $p_k$.

This is almost what we wanted to prove, except that $p_0$ is not a performance. Nonetheless, it is a function of the performances: specifically, a weighted average of historical ratings which, using this same lemma as an inductive hypothesis, are themselves monotonic in past performances. By induction, the proof is complete.
\end{proof}

Finally, we conclude that the player's incentives are aligned with optimizing round rankings, or raw scores:

\begin{theorem}[Aligned Incentives]
\label{thm:mono}
Holding fixed the set of contest rounds in which each player has participated, and the historical ratings and relative rankings of all players other than $i$, player $i$'s current rating is monotonic in each of their past rankings.
\end{theorem}

\begin{proof}
Choose any contest round in player $i$'s history, and consider improving player $i$'s rank in that round while holding everything else fixed. It suffices to show that player $i$'s current rating would necessarily increase as a result.

In the altered round, by \Cref{lem:mono-perf}, $p_i$ is increased; and by \Cref{lem:mono-rate}, player $i$'s post-round rating is increased. By \Cref{lem:mono-perf} again, this increases player $i$'s performance score in the following round. Proceeding inductively, we find that performance scores and ratings from this point onward are all increased.
\end{proof}

In the special cases of Elo-MM$\chi$ or Elo-MMR($\infty$), the rating system is ``memoryless'': the only data retained for each player are the current rating $\mu_{i,t}$ and uncertainty $\sigma_{i,t}$; detailed performance history is not saved. In this setting, we present a natural monotonicity theorem. A similar theorem was stated for the Codeforces system in \cite{Codeforces}, but no proofs were given.

\begin{theorem}[Memoryless Monotonicity Theorem]
In either the Elo-MM$\chi$ or Elo-MMR($\infty$) system, suppose $i$ and $j$ are two participants of round $t$. Suppose that the ratings and corresponding uncertainties satisfy $\mu_{i,t-1} \ge \mu_{j,t-1},\; \sigma_{i,t-1} = \sigma_{j,t-1}$. Then, $\sigma_{i,t} = \sigma_{j,t}$. Furthermore, if $i \succ j$ in round $t$, then $\mu_{i,t} > \mu_{j,t}$. On the other hand, if $j \succ i$ in round $t$, then $\mu_{j,t} - \mu_{j,t-1} > \mu_{i,t} - \mu_{i,t-1}$.
\end{theorem}

\begin{proof}
The new contest round will add a rating perturbation with variance $\gamma_t^2$, followed by a new performance with variance $\beta_t^2$. As a result,
\[\sigma_{i,t}
= \left( \frac{1}{\sigma_{i,t-1}^2 + \gamma_t^2} + \frac{1}{\beta_t^2} \right)^{-\frac 12}
= \left( \frac{1}{\sigma_{j,t-1}^2 + \gamma_t^2} + \frac{1}{\beta_t^2} \right)^{-\frac 12}
= \sigma_{j,t}.\]

The remaining conclusions are consequences of three properties: memorylessness, aligned incentives (\Cref{thm:mono}), and translation-invariance (ratings, skills, and performances are quantified on a common interval scale relative to one another).

Since the Elo-MM$\chi$ or Elo-MMR($\infty$) systems are memoryless, we may replace the initial prior and performance histories of players with any alternate histories of our choosing, as long as our choice is compatible with their current rating and uncertainty. For example, both $i$ and $j$ can be considered to have participated in the same set of rounds, with $i$ always performing at $\mu_{i,t-1}$. and $j$ always performing at $\mu_{j,t-1}$. Round $t$ is unchanged.

Suppose $i \succ j$. Since $i$'s historical performances are all equal or stronger than $j$'s, \Cref{thm:mono} implies $\mu_{i,t} > \mu_{j,t}$.

Suppose $j \succ i$. By translation-invariance, if we shift each of $j$'s performances, up to round $t$ and including the initial prior, upward by $\mu_{i,t-1} - \mu_{j,t-1}$, the rating changes between rounds will be unaffected. Players $i$ and $j$ now have identical histories, except that we still have $j\succ i$ at round $t$. Therefore, $\mu_{j,t-1} = \mu_{i,t-1}$ and, by \Cref{thm:mono}, $\mu_{j,t} > \mu_{i,t}$. Subtracting the equation from the inequality proves the second conclusion.
\end{proof}

\subsection{Robust response}
\label{sec:robust}

Another desirable property in many settings is robustness: a player's rating should not change too much in response to any one contest, no matter how extreme their performance. The Codeforces and TrueSkill systems lack this property, allowing for unbounded rating changes. TopCoder achieves robustness by clamping any changes that exceed a cap, which is initially high for new players but decreases with experience.

When $\rho>0$, Elo-MMR($\rho$) achieves robustness in a natural, smoother manner. It comes out of the interplay between Gaussian and logistic factors in the posterior; $\rho>0$ ensures that the Gaussian contribution doesn't vanish. Recall the notation used to describe the general posterior in \Cref{eq:posterior,eq:loss}, enhanced with the fractional multiplicities $\omega_k$ from \Cref{sec:pseudodiffusion}.

\begin{theorem}
\label{thm:robust}
In the Elo-MMR($\rho$) rating system, let
\[\Delta_{+} := \lim_{p_{t}\rightarrow+\infty} \mu_{t}-\mu_{t-1}, \quad \Delta_{-} := \lim_{p_{t}\rightarrow-\infty}\mu_{t-1}-\mu_{t}.
\]
Then,
\[\frac{\pi}{\beta_t\sqrt 3}
\left(\frac{1}{\beta_0^2} + \frac{\pi^2}{6}\sum_{k\in\cH_{t-1}}\frac{\omega_k}{\beta_k^2} \right)^{-1}
\le \Delta_{\pm} \le \frac{\pi\beta_0^2}{\beta_t\sqrt 3}.\]
\end{theorem}

\begin{proof}
Using the fact that $0 < \frac{d}{dx}\tanh(x) \le 1$, differentiating \Cref{eq:loss} yields
\[\frac{1}{\beta_0^2} \le L''(s)
\le \frac{1}{\beta_0^2} + \frac{\pi^2}{6}\sum_{k\in\cH_{t-1}}\frac{\omega_k}{\beta_k^2}.\]

For every $s\in\mathbb R$, in the limit as $p_t\rightarrow\pm\infty$, the new term corresponding to the performance at round $t$ will increase $L'(s)$ by $\mp\frac{\pi}{\beta_t\sqrt 3}$. Since $\mu_{t-1}$ was a zero of $L'$ without this new term, we now have
$L'(\mu_{t-1}) \rightarrow \mp\frac{\pi}{\beta_t\sqrt 3}.$ Dividing by the former inequalities yields the desired result.
\end{proof}

The proof reveals that the magnitude of $\Delta_{\pm}$ depends inversely on that of $L''$ in the vicinity of the current rating, which in turn is related to the derivative of the $\tanh$ terms. If a player's performances vary wildly, then most of the $\tanh$ terms will be in their tails, which contribute small derivatives, enabling larger rating changes. Conversely, the $\tanh$ terms of a player with a very consistent rating history will contribute large derivatives, so the bound on their rating change will be small.

Thus, Elo-MMR($\rho$) naturally caps the rating change of all players, and puts a smaller cap on the rating change of consistent players. The cap will increase after an extreme performance, providing a similar ``momentum'' to the TopCoder and Glicko-2 systems, but without sacrificing aligned incentives (\Cref{thm:mono}).

By comparing against \Cref{eq:variance}, we see that the lower bound in \Cref{thm:robust} is on the order of $\sigma_t^2/\beta_t$, while the upper bound is on the order of $\beta_0^2/\beta_t$. As a result, the momentum effect is more pronounced when $\beta_0$ is much larger than $\sigma_t$. Since the decay step increases $\beta_0$ while the transfer step decreases it, this occurs when the transfer rate $\rho$ is comparatively small. Thus, $\rho$ can be chosen in inverse proportion to the desired strength of momentum.

\subsection{Runtime analysis and optimizations}
\label{sec:runtime}
Let's look at the computation time needed to process a round with participant set $\mathcal P$, where we again omit the round subscript. Each player $i$ has a participation history $\cH_i$.

Estimating $P_i$ entails finding the zero of a monotonic function with $O(|\mathcal P|)$ terms, and then obtaining the rating $\mu_i$ entails finding the zero of another monotonic function with $O(|\cH_i|)$ terms. Using the Illinois or Newton methods, solving these equations to precision $\epsilon$ takes $O(\log\log\frac 1\epsilon)$ iterations. As a result, the total runtime needed to process one round of competition is
\[O\left(\sum_{i\in\mathcal P}(|\mathcal P| + |\cH_i|) \log\log\frac 1\epsilon\right).\]
This complexity is more than adequate for Codeforces-style competitions with thousands of contestants and history lengths up to a few hundred. Indeed, we were able to process the entire history of Codeforces on a small laptop in less than half an hour. Nonetheless, it may be cost-prohibitive in truly massive settings, where $|\mathcal P|$ or $|\cH_i|$ number in the millions. Fortunately, it turns out that both functions may be compressed down to a bounded number of terms, with negligible loss of precision.

\paragraph{Adaptive subsampling}
In \Cref{sec:bayes_model}, we used Doob's consistency theorem to argue that our estimate for $P_i$ is consistent. Specifically, we saw that $O(1/\epsilon^2)$ opponents are needed to get the typical error below $\epsilon$. Thus, we can subsample the set of opponents to include in the estimation, omitting the rest. Random sampling is one approach. A more efficient approach chooses a fixed number of opponents whose ratings are closest to that of player $i$, as these are more likely to provide informative match-ups. On the other hand, if the setting requires aligned incentives to hold exactly, then one must avoid choosing different opponents for each player.

\paragraph{History compression}
Similarly, it's possible to bound the number of stored factors in the posterior. Our skill-evolution algorithm decays the weights of old performances at an exponential rate. Thus, the contributions of all but the most recent $O(\log\frac 1\epsilon)$ terms are negligible. Rather than erase the older logistic terms outright, we recommend replacing them with moment-matched Gaussian terms, similar to the transfers in \Cref{sec:skill-drift} with $\kappa=0$. Since Gaussians compose easily, a single term can then summarize an arbitrarily long prefix of the history.

Substituting $1/\epsilon^2$ and $\log\frac 1\epsilon$ for $|\cP|$ and $|\cH_i|$, respectively, the runtime of Elo-MMR with both optimizations becomes
\[O\left(\frac {|\mathcal P|}{\epsilon^2} \log\log\frac 1\epsilon\right).\]

Finally, we note that the algorithm is embarrassingly parallel, with each player able to solve its equations independently. The threads can read the same global data structures, so each additional thread only contributes $O(1)$ memory overhead.

\section{Experiments}
\label{sec:experiments}
In this section, we compare various rating systems on real-world datasets, mined from several sources that will be described in \Cref{sec:datasets}. The metrics are runtime and predictive accuracy, as described in \Cref{sec:metrics}.

We compare Elo-MM$\chi$ and Elo-MMR($\rho$) against the industry-tested rating systems of Codeforces and TopCoder. For a fairer comparison, we hand-coded efficient versions of all four algorithms in the safe subset of Rust, parellelized using the Rayon crate; as such, the Rust compiler verifies that they contain no data races~\cite{stone2017rayon}. Our implementation of Elo-MMR($\rho$) makes use of the optimizations in \Cref{sec:runtime}, bounding both the number of sampled opponents and the history length by 500. In addition, we test the improved TrueSkill algorithm of \cite{NS10}, basing our code on an open-source implementation of the same algorithm. The inherent seqentiality of its message-passing procedure prevented us from parallelizing it.

\paragraph{Hyperparameter search}
To ensure fair comparisons, we ran a separate grid search for each triple of algorithm, dataset, and metric, over all of the algorithm's hyperparameters. The hyperparameter set that performed best on the first 10\% of the dataset, was then used to test the algorithm on the remaining 90\% of the dataset. 

The experiments were run on a 2.0 GHz 24-core Skylake machine with 24 GB of memory. Implementations of all rating systems, hyperparameters, datasets, and additional processing used in our experiments can be found at \url{https://github.com/EbTech/EloR/}.


\subsection{Datasets}
\label{sec:datasets}

\begin{table}[t]
\begin{tabular}{l|l|l}
\hline
\textbf{Dataset} & \textbf{\# contests} & \textbf{avg. \# participants / contest} \\ \hline
Codeforces       & 1087                & 2999                                     \\ 
TopCoder         & 2023                & 403                                   \\ 
Reddit           & 1000                & 20                                       \\
Synthetic        & 50                  & 2500     \\ \hline
\end{tabular}
    \caption{Summary of test datasets.}
    \label{tab:dataset-summary}
    \vspace{-1.2em}
\end{table}

Due to the scarcity of public domain datasets for rating systems, we mined three datasets to analyze the effectiveness of our system. The datasets were mined using data from each source website's inception up to October 12th, 2020. We also created a synthetic dataset to test our system's performance when the data generating process matches our theoretical model. Summary statistics of the datasets are presented in \Cref{tab:dataset-summary}.

\paragraph{Codeforces contest history}
This dataset contains the current entire history of rated contests ever run on CodeForces.com, the dominant platform for online programming competitions. The CodeForces platform has over 850K users, over 300K of whom are rated, and has hosted over 1000 contests to date. Each contest has a couple thousand competitors on average. A typical contest takes 2 to 3 hours and contains 5 to 8 problems. Players are ranked by total points, with more points typically awarded for tougher problems and for early solves. They may also attempt to ``hack'' one another's submissions for bonus points, identifying test cases that break their solutions. 
\looseness=-1

\paragraph{TopCoder contest history}
This dataset contains the current entire history of algorithm contests ever run on the TopCoder.com. TopCoder is a predecessor to Codeforces, with over 1.4 million total users and a long history as a pioneering platform for programming contests. It hosts a variety of contest types, including over 2000 algorithm contests to date. The scoring system is similar to Codeforces, but its rounds are shorter: typically 75 minutes with 3 problems.

\paragraph{SubRedditSimulator threads}
This dataset contains data scraped from the current top-1000 most upvoted threads on the website \url{reddit.com/r/SubredditSimulator/}. Reddit is a social news aggregation website with over 300 million users. The site itself is broken down into sub-sites called subreddits. Users then post and comment to the subreddits, where the posts and comments receive votes from other users. In the subreddit SubredditSimulator, users are language generation bots trained on text from other subreddits. Automated posts are made by these bots to SubredditSimulator every 3 minutes, and real users of Reddit vote on the best bot. Each post (and its associated comments) can thus be interpreted as a round of competition between the bots who commented. 

\paragraph{Synthetic data}
This dataset contains 10K players, with skills and performances generated according to the Gaussian generative model in \Cref{sec:bayes_model}. Players' initial skills are drawn i.i.d. with mean $1500$ and variance $300$. Players compete in all rounds, and are ranked according to independent performances with variance $200$. Between rounds, we add i.i.d. Gaussian increments with variance $35$ to each of their skills.

\subsection{Evaluation metrics}
\label{sec:metrics}
To compare the different algorithms, we define two measures of predictive accuracy. Each metric will be defined on individual contestants in each round, and then averaged:
\[\mathrm{\bf aggregate(metric)} := \frac{\sum_t \sum_{i\in\mathcal P_t} \mathrm{\bf metric}(i,t)}{\sum_t |\mathcal P_t|}.\]

\paragraph{Pair inversion metric~\cite{HMG06}}
Our first metric computes the fraction of opponents against whom our ratings predict the correct pairwise result, defined as the higher-rated player either winning or tying: 
\[\mathrm{\bf pair\_inversion}(i,t) := \frac{\text{\# correctly predicted matchups}}{|\mathcal P_t|-1} \times 100\%.\]
This metric was used in the evaluation of TrueSkill~\cite{HMG06}.

\paragraph{Rank deviation}
Our second metric compares the rankings with the total ordering that would be obtained by sorting players according to their prior rating. The penalty is proportional to how much these ranks differ for player $i$:
\[\mathrm{\bf rank\_deviation}(i,t) := \frac{|\text{actual\_rank} - \text{predicted\_rank}|}{|\mathcal P_t|-1} \times 100\%.\]
In the event of ties, among the ranks within the tied range, we use the one that comes closest to the rating-based prediction.


\subsection{Empirical results}
\begin{table*}
\begin{tabular}{l|ll|ll|ll|ll|ll}
 \hline
\multirow{2}{*}{\textbf{Dataset}} &
  \multicolumn{2}{l|}{\textbf{Codeforces}} &
  \multicolumn{2}{l|}{\textbf{TopCoder}} &
  \multicolumn{2}{l|}{\textbf{TrueSkill}} &
  \multicolumn{2}{l|}{\textbf{Elo-MM$\boldsymbol\chi$}} & 
  \multicolumn{2}{l}{\textbf{Elo-MMR($\boldsymbol\rho$)}} \\ \cline{2-11}
&
  pair inv. &
  rank dev. &
  pair inv. &
  rank dev. &
  pair inv. &
  rank dev. &
  pair inv. &
  rank dev. &
  pair inv. &
  rank dev. \\ \hline
Codeforces & 78.3\% & 14.9\% & 78.5\% & 15.1\% & 61.7\% & 25.4\% & 78.5\% & 14.8\% & {\bf 78.6}\% & {\bf 14.7}\% \\ 
TopCoder  & 72.6\%     & 18.5\%     & 72.3\% & 18.7\%  & 68.7\% & 20.9\% & 73.0\% & 18.3\% & {\bf 73.1}\% & {\bf 18.2}\% \\ 
Reddit     & 61.5\%     & 27.3\%     & 61.4\% & 27.4\% & 61.5\% & {\bf 27.2}\% & 61.6\% & 27.3\% & {\bf 61.6\%} & 27.3\% \\ 
Synthetic  & {\bf 81.7\%}     & 12.9\%     & {\bf 81.7}\% & {\bf 12.8}\% & 81.3\% & 13.1\% & {\bf 81.7}\% & {\bf 12.8}\% & {\bf 81.7\%} & {\bf 12.8\%} \\ \hline
\end{tabular}
\caption{Performance of each rating system on the pairwise inversion and rank deviation metrics. Bolded entries denote the best performances (highest pair inv. or lowest rank dev.) on each metric and dataset.}
\label{tbl:metric-results}
\vspace{-1.2em}
\end{table*}

\begin{table}
\begin{tabular}{l|lllll}
\hline
\textbf{Dataset} & \textbf{CF} & \textbf{TC} & \textbf{TS} & \textbf{Elo-MM$\boldsymbol\chi$} & \textbf{Elo-MMR($\boldsymbol\rho$)} \\ \hline
Codeforces & 212.9 & 72.5 & 67.2 & {\bf 31.4} & 35.4\\
TopCoder   & 9.60 & {\bf 4.25} & 16.8 & 7.00 & 7.52\\
Reddit     & 1.19  & 1.14 & {\bf 0.44} & 1.14 & 1.42 \\
Synthetic  & 3.26  & 1.00 & 2.93 & {\bf 0.81} & 0.85 \\ \hline
\end{tabular}
\caption{Total compute time over entire dataset, in seconds.}
\label{tbl:time-results}
\vspace{-1.2em}
\end{table}

Recall that Elo-MM$\chi$ has a Gaussian performance model, matching the modeling assumptions of TopCoder and TrueSkill. Elo-MMR($\rho$), on the other hand, has a logistic performance model, matching the modeling assumptions of Codeforces and Glicko. While $\rho$ was included in the hyperparameter search, in practice we found that all values between $0$ and $1$ produce very similar results.

To ensure that errors due to the unknown skills of new players don't dominate our metrics, we excluded players who had competed in less than 5 total contests. In most of the datasets, this reduced the performance of our method relative to the others, as our method seems to converge more accurately. Despite this, we see in \Cref{tbl:metric-results} that both versions of Elo-MMR outperform the other rating systems in both the pairwise inversion metric and the ranking deviation metric.
\looseness=-1

We highlight a few key observations. First, significant performance gains are observed on the Codeforces and TopCoder datasets, despite these platforms' rating systems having been designed specifically for their needs. Our gains are smallest on the synthetic dataset, for which all algorithms perform similarly. This might be in part due to the close correspondence between the generative process and the assumptions of these rating systems. Furthermore, the synthetic players compete in all rounds, enabling the system to converge to near-optimal ratings for every player. Finally, the improved TrueSkill performed well below our expectations, despite our best efforts to improve it; we suspect that the message-passing algorithm breaks down in contests with a large number of distinct ranks. To our knowledge, we are the first to present experiments with TrueSkill on contests where the number of distinct ranks is in the hundreds or thousands. In preliminary experiments, TrueSkill and Elo-MMR score about equally when the number of ranks is less than about 60.

Now, we turn our attention to \Cref{tbl:time-results}, which showcases the computational efficiency of Elo-MMR. On smaller datasets, it performs comparably to the Codeforces and TopCoder algorithms. However, the latter suffer from a quadratic time dependency on the number of contestants; as a result, Elo-MMR outperforms them by almost an order of magnitude on the larger Codeforces dataset.

Finally, in comparisons between the two Elo-MMR variants, we note that while Elo-MMR($\rho$) is more accurate, Elo-MM$\chi$ is always faster. This has to do with the skill drift modeling described in \Cref{sec:skill-drift}, as every update in Elo-MMR($\rho$) must process $O(\log\frac 1\epsilon)$ terms of a player's competition history.

\section{Conclusions}
This paper introduces the Elo-MMR rating system, which is in part a generalization of the two-player Glicko system, allowing an unbounded number of players. By developing a Bayesian model and taking the limit as the number of participants goes to infinity, we obtained simple, human-interpretable rating update formulas. Furthermore, we saw that the algorithm is asymptotically fast, embarrassingly parallel, robust to extreme performances, and satisfies the important \emph{aligned incentives} property. To our knowledge, our system is the first to rigorously prove all these properties in a setting with more than two individually ranked players. In terms of practical performance, we saw that it outperforms existing industry systems in both prediction accuracy and computation speed.

This work can be extended in several directions. First, the choices we made in modeling ties, pseudodiffusions, and opponent subsampling are by no means the only possibilities consistent with our Bayesian model of skills and performances. Second, one may obtain better results by fitting the performance and skill evolution models to application-specific data.

Another useful extension would be to team competitions. While it's no longer straightforward to infer precise estimates of an individual's performance, Elo-MM$\chi$ can simply be applied at the team level. To make this useful in settings where players may form new teams in each round, we must model teams in terms of their individual members. In the case where a team's performance is modeled as the sum of its members' independent Gaussian contributions, elementary facts about multivariate Gaussian distributions enable posterior skill inferences at the individual level. Generalizing this approach remains an open challenge.


Over the past decade, online competition communities such as Codeforces have grown exponentially. As such, considerable work has gone into engineering scalable and reliable rating systems. Unfortunately, many of these systems have not been rigorously analyzed in the academic community. We hope that our paper and open-source release will open new explorations in this area.



\section*{Acknowledgements}
The authors are indebted to Daniel Sleator and Danica J. Sutherland for initial discussions that helped inspire this work, and to Nikita Gaevoy for the open-source improved TrueSkill upon which our implementation is based. Experiments in this paper were funded by a Google Cloud Research Grant. The second author is supported by a VMWare Fellowship and the Natural Sciences and Engineering Research Council of Canada.
\looseness=-1

\appendix
\section*{Appendix}
\begingroup
\def\thetheorem{\ref{lem:decrease}}
\begin{lemma}
If $f_i$ is continuously differentiable and log-concave, then the functions $l_i,d_i,v_i$ are continuous, strictly decreasing, and
\[l_i(p) < d_i(p) < v_i(p) \text{ for all }p.\]
\end{lemma}
\addtocounter{theorem}{-1}
\endgroup
\begin{proof}
Continuity of $F_i,f_i,f'_i$ implies that of $l_i,d_i,v_i$. It's known~\cite{concave} that log-concavity of $f_i$ implies log-concavity of both $F_i$ and $1-F_i$. As a result, $l_i$, $d_i$, and $v_i$ are derivatives of strictly concave functions; therefore, they are strictly decreasing. In particular, each of

\[v'_i(p) = \frac{f'_i(p)}{F_i(p)} - \frac{f_i(p)^2}{F_i(p)^2},\quad
l'_i(p) = \frac{-f'_i(p)}{1-F_i(p)} - \frac{f_i(p)^2}{(1-F_i(p))^2},\]

are negative for all $p$, so we conclude that

\begin{align*}
d_i(p) - v_i(p)
= \frac{f'_i(p)}{f_i(p)} - \frac{f_i(p)}{F_i(p)}
&= \frac{F_i(p)}{f_i(p)} v'_i(p)
< 0,
\\l_i(p) - d_i(p)
= -\frac{f'_i(p)}{f_i(p)} -\frac{f_i(p)}{1-F_i(p)}
&= \frac{1-F_i(p)}{f_i(p)} l'_i(p)
< 0.
\end{align*}

\end{proof}

\begingroup
\def\thetheorem{\ref{thm:uniq-max}}
\begin{theorem}
Suppose that for all $j$, $f_j$ is continuously differentiable and log-concave. Then the unique maximizer of $\Pr(P_i=p\mid E^L_i,E^W_i)$ is given by the unique zero of
\[Q_i(p) = \sum_{j \succ i} l_j(p) + \sum_{j \sim i} d_j(p) + \sum_{j \prec i} v_j(p).\]
\end{theorem}
\addtocounter{theorem}{-1}
\endgroup

\begin{proof}
First, we rank the players by their buckets according to $\floor{P_j/\epsilon}$, and take the limiting probabilities as $\epsilon\rightarrow 0$:
\begin{align*}
    \Pr(\floor{\frac{P_j}\epsilon} > \floor{\frac{p}\epsilon})
    &= \Pr(p_j \ge \epsilon\floor{\frac{p}\epsilon} + \epsilon)
    \\&= 1 - F_j(\epsilon\floor{\frac{p}\epsilon} + \epsilon)
    \rightarrow 1 - F_j(p),
    \\\Pr(\floor{\frac{P_j}\epsilon} < \floor{\frac{p}\epsilon})
    &= \Pr(p_j < \epsilon\floor{\frac{p}\epsilon})
    \\&= F_j(\epsilon\floor{\frac{p}\epsilon})
    \rightarrow F_j(p),
    \\\frac 1\epsilon \Pr(\floor{\frac{P_j}\epsilon} = \floor{\frac{p}\epsilon})
    &= \frac 1\epsilon \Pr(\epsilon\floor{\frac{p}\epsilon} \le P_j < \epsilon\floor{\frac{p}\epsilon} + \epsilon)
    \\&= \frac 1\epsilon\left( F_j(\epsilon\floor{\frac{p}\epsilon} + \epsilon) - F_j(\epsilon\floor{\frac{p}\epsilon}) \right)
    \rightarrow f_j(p).
\end{align*}

Let $L_{jp}^\epsilon$, $W_{jp}^\epsilon$, and $D_{jp}^\epsilon$ be shorthand for the events $\floor{\frac{P_j}\epsilon} > \floor{\frac{p}\epsilon}$, $\floor{\frac{P_j}\epsilon} < \floor{\frac{p}\epsilon}$, and $\floor{\frac{P_j}\epsilon} = \floor{\frac{p}\epsilon}$. respectively. These correspond to a player who performs at $p$ losing, winning, and drawing against $j$, respectively, when outcomes are determined by $\epsilon$-buckets. Then,
\begin{align*}
\Pr(E^W_i,E^L_i\mid P_i=p)
&= \lim_{\epsilon\rightarrow 0}
\prod_{j \succ i} \Pr(L_{jp}^\epsilon)
\prod_{j \prec i} \Pr(W_{jp}^\epsilon)
\prod_{j \sim i, j\ne i} \frac{\Pr(D_{jp}^\epsilon)}\epsilon
\\&= \prod_{j \succ i} (1 - F_j(p)) \prod_{j \prec i} F_j(p) \prod_{j \sim i, j\ne i} f_j(p),
\\\Pr(P_i=p \mid E^L_i,E^W_i)
&\propto f_i(p) \Pr(E^L_i,E^W_i\mid P_i=p)
\\&= \prod_{j \succ i} (1 - F_j(p)) \prod_{j \prec i} F_j(p) \prod_{j \sim i} f_j(p),
\\\ddp\ln \Pr(P_i=p \mid E^L_i,& E^W_i) = \sum_{j \succ i} l_j(p) + \sum_{j \prec i} v_j(p) + \sum_{j \sim i} d_j(p) = Q_i(p).
\end{align*}

Since \Cref{lem:decrease} tells us that $Q_i$ is strictly decreasing, it only remains to show that it has a zero. If the zero exists, it must be unique and it will be the unique maximum of $\Pr(P_i=p \mid E^L_i,E^W_i)$.

To start, we want to prove the existence of $p^*$ such that $Q_i(p^*) < 0$. Note that it's not possible to have $f'_j(p) \ge 0$ for all $p$, as in that case the density would integrate to either zero or infinity. Thus, for each $j$ such that $j\sim i$, we can choose $p_j$ such that $f'_j(p_j) < 0$, and so $d_j(p_j) < 0$. Let $\alpha = -\sum_{j\sim i} d_j(p_j) > 0$.

Let $n = |\{j:\,j \prec i\}|$. For each $j$ such that $j \prec i$, since $\lim_{p\rightarrow\infty}v_j(p) = 0/1 = 0$, we can choose $p_j$ such that $v_j(p_j) < \alpha/n$. Let $p^* = \max_{j\preceq i} p_j$. Then,
\[
\sum_{j \succ i} l_j(p^*) \le 0, \quad \sum_{j \sim i} d_j(p^*) \le -\alpha, \quad \sum_{j \prec i} v_j(p^*) < \alpha.
\]

Therefore,
\begin{align*}
Q_i(p^*)
&= \sum_{j \succ i} l_j(p^*) + \sum_{j \sim i} d_j(p^*) + \sum_{j \prec i} v_j(p^*)
\\&< 0 - \alpha + \alpha = 0.
\end{align*}

By a symmetric argument, there also exists some $q^*$ for which $Q_i(q^*) > 0$. By the intermediate value theorem with $Q_i$ continuous, there exists $p\in (q^*,p^*)$ such that $Q_i(p) = 0$, as desired.
\looseness=-1
\end{proof}

\bibliographystyle{ACM-Reference-Format}
\bibliography{EloR}

\end{document}